\documentclass[11pt]{article}

\usepackage[margin=1in]{geometry}
\usepackage{amsmath,amssymb,amsthm,mathtools}
\usepackage{microtype}
\usepackage{booktabs}
\usepackage{enumitem}
\usepackage{xcolor}
\usepackage[colorlinks=true,linkcolor=blue,citecolor=blue,urlcolor=blue]{hyperref}
\usepackage{setspace}
\newtheorem{theorem}{Theorem}
\newtheorem{lemma}{Lemma}
\newtheorem{claim}{Claim}
\newtheorem{proposition}{Proposition}

\newtheorem{definition}{Definition}

\numberwithin{equation}{section}

\newcommand{\E}{\mathbb E}
\newcommand{\Prb}{\mathbb P}
\newcommand{\OPT}{\mathrm{OPT}}
\newcommand{\MW}{\mathrm{MW}}

\newcommand{\Cmax}{C_{\max}}
\newcommand{\one}{\mathbf 1}

\newcommand{\DetBIC}{\mathrm{DetBIC}}
\newcommand{\RandBIC}{\mathrm{RandBIC}}

\title{Gaps and Augmentations in Bayesian Scheduling}
\author{Ahuva Mu'alem%
\thanks{The author is grateful to the WINE 2026 referees for their thoughtful comments and suggestions. Financial support from the Australian Research Council under grant DP190102064 is gratefully acknowledged.}}

\begin{document}
\maketitle

\begin{abstract}
The recent resolution of the Nisan--Ronen conjecture~\cite{NR,CKK} establishes that the optimal worst-case approximation ratio achievable by deterministic truthful mechanisms for minimizing makespan  on \(m\) unrelated machines is exactly \(m\). This motivates asking how prior information, Bayesian incentive compatibility (BIC), randomization, and machine-side resource augmentation change this prior-free barrier.

We obtain three main results. First, we prove an asymptotic \(4/3\) lower bound for randomized BIC mechanisms with two machines. This strengthens the previous \(1.2\) deterministic-BIC lower bound~\cite{MS} and extends the impossibility to randomized-BIC mechanisms. Second, randomization is already known to strictly outperform deterministic mechanisms in the prior-free setting: for two machines, a randomized truthful \(7/4\)-approximation is known, whereas every deterministic truthful mechanism has approximation ratio at least \(2\)~\cite{NR}. We show that randomized BIC scheduling mechanisms are likewise strictly more powerful than their deterministic counterparts, exhibiting an asymptotic BIC-integrality gap of \(8/7\). Thus, optimality in prior-dependent BIC scheduling can strictly require randomization. This parallels the role of lotteries in multidimensional revenue maximization, where randomization can strictly improve revenue~\cite{MV,BCKW}. 
Third, we show that when job assignments are sufficiently well spread across machines (more formally, when the pairwise collision parameter \(\Delta_r\) is bounded by a constant independent of both the number of machines \(m\) and the number of sampled layers \(r\)), then \(O(m/\varepsilon)\) sampled layers suffice for the standard truthful MinWork mechanism to achieve a \((1+\varepsilon)\)-approximation to the original first-best benchmark. Equivalently, \(O(m/\varepsilon)\) sampled replicas per machine suffice. This resource-augmentation result parallels the Bulow--Klemperer perspective~\cite{BK,EFFTW}. As a by-product, in the standard \emph{unaugmented} model, MinWork achieves a
\(
\left(1+\frac{\Delta_1(m-1)}{2}\right)\text{-approximation}
\)
to the first-best benchmark for every prior.
\end{abstract}

%%%%%%%%%%%%%%%%%%%%%%%%%%%%%%%%%%%%%%%%%%%%%%%%%%%%%%%%%%%%%%%%%
\section{Introduction}
Unrelated-machines makespan minimization is a canonical problem in algorithmic mechanism design. There are \(m\) machines and \(n\) jobs. Machine \(i\)'s private type is its vector of processing times for the \(n\) jobs, and the objective is to minimize the makespan, i.e., the maximum load over machines. If the processing-time matrix were publicly known, this would be a classical scheduling problem. The mechanism-design challenge is that each row of the matrix is privately held by a strategic machine.

The seminal work of Nisan and Ronen initiated the study of truthful unrelated-machines scheduling in algorithmic mechanism design~\cite{NR}. Their MinWork mechanism assigns each job to a machine reporting the smallest processing time for that job and uses VCG payments. It is truthful, meaning that reporting the true processing times is a dominant strategy for every machine. 
Nisan and Ronen showed that MinWork is an \(m\)-approximation, proved a lower bound of \(2\), and conjectured that no deterministic truthful mechanism could achieve an approximation ratio below \(m\); this conjecture was recently proved by Christodoulou, Koutsoupias, and Kov\'acs~\cite{CKK},  closing a long-standing open problem in prior-free algorithmic mechanism design.

This sharp prior-free barrier motivates asking how prior information, Bayesian incentive compatibility, randomization, and resource augmentation change the picture. With a known prior and BIC, truthful reporting maximizes each machine’s expected utility conditional on its own type. To what extent can stronger makespan guarantees be obtained?  Are randomized BIC scheduling mechanisms strictly more powerful than deterministic BIC mechanisms? Finally, can adding more machines help the simple truthful MinWork mechanism approach the original first-best benchmark, that is, the minimum makespan achievable when the realized processing times are known and incentive constraints are ignored?

\paragraph{A \(4/3\) Bayesian Lower Bound.}
Our first result is an asymptotic \(4/3\) lower bound for randomized BIC mechanisms.
We construct a family of two-machine, two-job, two-type-per-machine product priors for which the randomized-BIC-to-first-best ratio approaches \(4/3\).
This strengthens the 1.2 lower-bound of Mu'alem and Schapira for deterministic-BIC mechanisms~\cite{MS} in three aspects: it applies to randomized rather than only deterministic BIC mechanisms, the  asymptotic lower-bound constant increases from 1.2 to \(4/3\), and the hard instance uses only two jobs, rather than three.
As a warm-up, we show that the same hard instance already yields an asymptotic \(4/3\) lower bound for \emph{every} deterministic BIC mechanism. This class includes, in particular, all deterministic dominant-strategy truthful mechanisms. Consequently, the same \(4/3\) lower bound applies to weighted-VCG rules and affine minimizers, even when their report-independent weights and offsets are chosen as a function of the prior.~\footnote{In the usual value-maximization convention, these mechanisms are known as affine maximizers or weighted VCG mechanisms~\cite{Roberts1979}; in our cost-minimization setting, the corresponding rule is an affine minimizer.} 

To place the $4/3$ lower bound in context, prior-free truthful mechanisms provide nontrivial upper bounds that also apply in the Bayesian setting. For two machines and two jobs, the best known randomized truthful upper bound is  \(1.5059964\)~\cite{KV}.

\paragraph{The Value of Randomization: A BIC Integrality Gap.}
Randomization is already known to improve truthful approximation guarantees in the prior-free model. For two machines, Nisan and Ronen established a tight  bound of 2 for deterministic truthful mechanisms, while also giving a randomized truthful \(7/4\)-approximation~\cite{NR}. This raises a natural question: can randomization also strictly improve the optimal BIC makespan? 
Our second result answers this question affirmatively.
 We construct  a two-machine, two-job, two-type product-prior  that   yields
$ \frac{\OPT_{\DetBIC}}{\OPT_{\RandBIC}}\ge \frac87. $
Thus, randomization remains valuable even within the weaker BIC solution concept. The separation already appears in a simple environment, regardless of computational  considerations. 
This parallels the role of lotteries in multidimensional revenue maximization, where randomization can be essential for optimality and can strictly improve revenue~\cite{MV,BCKW}.~\footnote{Importantly, the deterministic-versus-randomized-BIC integrality gap and the randomized-BIC-versus-first-best approximation ratio measure different phenomena and need not be witnessed by the same prior.}

\paragraph{Resource augmentation for MinWork.}
Our third result is naturally framed in the language of competition
complexity~\cite{BK,EFFTW}. We ask how much additional machine-side
competition is needed for the simple truthful MinWork mechanism to approach
the first-best benchmark in the original \(m\)-machine market. We sample
\(r\) independent processing-time profiles from the original prior, creating
\(mr\) destination slots, and run MinWork in this replicated sampled environment.
Each slot is treated as a distinct strategic machine. Since MinWork minimizes
total reported processing cost, standard VCG payments make it
dominant-strategy truthful, independently of correlations in the prior.
The key quantity is how often two jobs end up at the same destination, which we measure by \(\Delta_r\). When these collisions are sufficiently rare, \(O(m/\varepsilon)\) layers (equivalently \(O(m/\varepsilon)\) additional replicas per machine) suffice for MinWork to achieve a \((1+\varepsilon)\)-approximation to the original first-best benchmark , and this dependence is asymptotically tight in the worst case.

This result is related to, but different from, the  Bulow--Klemperer-style scheduling result  of Chawla, Hartline, Malec, and Sivan~\cite{CHMS}. 
Their approach  first obtains a guarantee against a relaxed benchmark with fewer machines and then converts this guarantee back to the original 
\(m\)-machine benchmark under additional assumptions, such as MHR.  
Our analysis instead compares replicated MinWork directly with the original \(m\)-machine first best, so no MHR assumption is required.

Finally, setting \(r=1\) in our theorem gives an upper bound for  MinWork that holds for every prior: its expected approximation ratio is at most
$1+\frac{\Delta_1(m-1)}{2}$. This complements Giannakopoulos and Kyropoulou~\cite{GK}: their sharper analysis relies on stronger independence assumptions, whereas our bound applies more generally. Under their independent machine-identical setting, uniform tie-breaking gives \(\Delta_1=1\),
so our bound becomes \((m+1)/2\), which is tight  for two machines. 
For larger \(m\), their bound can be substantially sharper.

%%%%%%%%%%%%%%%%%%%%%%%%%%%%%%%%%%%%%%%%%%%%%%%%%%%%%%%%%%%%%%%%%
\subsection{Related Work}
Nisan and Ronen~\cite{NR} introduced unrelated-machines scheduling as a canonical problem in algorithmic mechanism design, asking how well makespan can be minimized while preserving truthfulness. In the standard prior-free model, no distribution over processing times is assumed. 
The long-standing Nisan--Ronen conjecture was recently resolved by Christodoulou, Koutsoupias, and Kov\'acs~\cite{CKK}, establishing that, for \(m\) unrelated machines, no deterministic dominant-strategy truthful mechanism can achieve an approximation ratio below \(m\).

 Randomization has long played an important role in truthful scheduling. Nisan and Ronen gave a randomized truthful 7/4-approximation for two unrelated machines~\cite{NR}; their paper credits Daniel Lehmann with the refined case analysis establishing the 7/4 bound. Lu and Yu improved this to 1.6737 and also obtained a 0.8368\(m\)-approximation for general \(m\)~\cite{LY}. Chen, Du, and Zuluaga developed copula-based randomized truthful mechanisms~\cite{CDZ}. They obtain an approximation ratio of at most \(1.58606\) for two machines in general and \(1.5067711\) for the two-machine, two-job case. Kuryatnikova and Vera subsequently improved the two-job upper bound to \(1.5059964\) and obtained a nearly matching lower bound within their randomized scale-free mechanism class~\cite{KV}.  Lavi and Swamy study a restricted domain in which every processing time is either low or high and  give a black-box transformation from any \(c\)-approximation algorithm to a truthful-in-expectation \(3c\)-approximation mechanism~\cite{LSched}. On the lower-bound side, Mu'alem and Schapira~\cite{MS} proved a \(2-1/m\) lower bound for prior-free randomized truthful scheduling. Their lower bound holds even for the weaker requirement of truthfulness in expectation.

Daskalakis and Weinberg study Bayesian truthful scheduling and give polynomial-time Bayesian mechanisms within a factor 2 of the optimal BIC mechanism~\cite{DW}. 
Their benchmark is the optimal BIC mechanism, whereas our  bounds measure the gap between BIC and the unconstrained first best.

The classical Bulow--Klemperer theorem shows that in single-item i.i.d. regular environments, second-price auction  with one additional bidder earns at least as much expected revenue as the optimal revenue-maximizing Bayesian auction with the original bidders~\cite{BK}.  Eden et al.~\cite{EFFTW} formalized this perspective through competition complexity, the number of additional bidders required for VCG to match or exceed the revenue of the optimal Bayesian mechanism with the original bidders.

Chawla, Hartline, Malec, and Sivan~\cite{CHMS} study prior-independent truthful scheduling under machine symmetry. For machine-symmetric product priors with MHR runtimes, their bounded-overload BIC mechanism achieves an \(O(n/m)\)-approximation, without requiring knowledge of the processing-time distributions.
Giannakopoulos and Kyropoulou~\cite{GK} take a different approach and analyze the Bayesian performance of MinWork: each job is assigned to the machine on which its processing time is smallest. They assume independent processing times that are identically distributed across machines for each job. Under this symmetry, the minimizing machine of each job is uniformly distributed, turning MinWork into a weighted balls-into-bins process. Using this connection, they prove an
$ O\!\left(\frac{\log m}{\log\log m}\right) $
approximation for \(m\) machines. 

%%%%%%%%%%%%%%%%%%%%%%%%%%%%%%%%%%%%%%%%%%%%%%%%%%%%%%%%%%%%%%%%%
\section{Model and Preliminaries}
\label{sec:model} 
For a positive integer \(q\), let \([q]=\{1,\ldots,q\}\).
There are \(m \ge 2 \) machines and \(n \ge 1 \) jobs. Let
\(
T=(T_{ij})_{i\in[m],\,j\in[n]}
\)
denote the random processing-time matrix, where \(T_{ij}\ge0\) is the
processing time of job \(j\) on machine \(i\).  Machine \(i\)'s random type is
its row \(T_i=(T_{ij})_{j\in[n]}\); we write \(t\)  for a realization of \(T\) and \(t_i\) for its \(i\)-th row.

For the lower-bound results in Sections~\ref{sec:four-thirds-randomized-bic-lb} and~\ref{sec:det-rand-separation}, 
the known prior is a product distribution,
\(
F=F_1\times\cdots\times F_m,
\)
where \(F_i\) is a distribution over \(\mathbb R_+^n\). 
Thus machine types are independent across machines, 
while coordinates within a machine's type vector may be arbitrarily correlated across jobs. Section~\ref{sec:resource-augmentation} instead  allows  an arbitrary  joint prior over \(T\). 

A randomized Bayesian mechanism consists of an allocation rule and payments.
Given a reported profile \(\widehat t\), it outputs a possibly randomized allocation and a payment vector
\((X(\widehat t),p(\widehat t))\), where every realization of \(X(\widehat t)\)
assigns each job to exactly one machine. For a realized allocation \(x\),
let \(x_j\) denote the machine receiving job \(j\). The load of machine \(i\)
is
\(
L_i(x,t_i)=\sum_{j=1}^n t_{ij}\one\{x_j=i\},
\)
and the makespan is
\(
\Cmax(x,t)=\max_{i\in[m]}L_i(x,t_i).
\)
The first-best optimum is
\(
\OPT_m(t)=\min_x\Cmax(x,t),
\)
where the minimum is over deterministic allocations.

\begin{lemma}[Randomization does not improve the first-best optimum]
\label{lem:randomized-opt}
For every realized processing-time matrix \(t\),
\[
\min_{\substack{\text{randomized}\\ \text{allocations }X}}
\E_X[\Cmax(X,t)]
=
\OPT_m(t).
\]
\end{lemma}

\begin{proof}
A randomized allocation is a distribution over deterministic allocations, so its expected makespan cannot be smaller than \(\OPT_m(t)\). Conversely, an optimal deterministic allocation is a degenerate randomized allocation and achieves \(\OPT_m(t)\).
\end{proof}

Accordingly, our benchmark is the expected first best, \(\E_{T\sim F}[\OPT_m(T)]\).
A mechanism is Bayesian incentive compatible (BIC) if truthful reporting maximizes every machine's interim expected utility. Formally, for every machine \(i\), every true realized type \(t_i\) and every alternative report \(\widehat t_i\), let \(T_{-i}\) denote the types of all machines other than machine \(i\). Then
\begin{align}
\E\!\left[p_i(t_i,T_{-i})-L_i(X(t_i,T_{-i}),t_i)\right]
\ge
\E\!\left[p_i(\widehat t_i,T_{-i})-L_i(X(\widehat t_i,T_{-i}),t_i)\right].
\label{eq:bic-definition}
\end{align}
The expectation is over the conditional distribution of the other machines' types given machine \(i\)'s true type, and over the internal randomness of the mechanism.  Payments are quasilinear: utility equals payment minus processing cost. 

For priors with \(\E[\OPT_m(T)]>0\), the approximation ratio of a BIC mechanism
\(M=(X,p)\), evaluated under truthful reporting, is
\(
\frac{\E[\Cmax(X(T), \; T)]}{\E[\OPT_m(T)]},
\)
where the expectation in the numerator is also taken over the mechanism's internal randomness.

The MinWork allocation rule assigns each job to a machine with minimum reported processing time. Equivalently, it minimizes the total reported processing time over all allocations. Therefore, together with the standard VCG payments, MinWork is dominant-strategy truthful (hence BIC) and ex-post individually rational~\cite{NR}.~\footnote{A mechanism is ex-post individually rational (IR) if truthful reporting gives every machine nonnegative utility for every realized type profile  and every realization of the mechanism's randomness.} No tie-breaking convention is needed for our lower bounds. In
Section~\ref{sec:resource-augmentation}, ties are broken uniformly among
minimizing destinations, independently across jobs, using report-independent
randomness. Our lower bounds use only BIC and therefore continue to hold when  individual rationality is additionally imposed.

%%%%%%%%%%%%%%%%%%%%%%%%%%%%%%%%%%%%%%%%%%%%%%%%%%%%%%%%%%%%%%%%%%%%%%%%%%%%%%
\section{The \texorpdfstring{$4/3$}{4/3} Randomized-BIC Lower Bound}
\label{sec:four-thirds-randomized-bic-lb}

This section establishes our main impossibility theorem. Moving from prior-free randomized truthfulness to Bayesian incentive compatibility can potentially improve performance, but our construction shows that even randomized BIC mechanisms still face an asymptotic \(4/3\) gap from the first-best benchmark. The hard instance is a product prior with only two machines, two jobs, and two types per machine.
This strengthens the 1.2 deterministic-BIC lower bound of Mu'alem and Schapira~\cite{MS} in three aspects: the impossibility applies to randomized rather than deterministic BIC mechanisms, raises the asymptotic constant from \(1.2\) to \(4/3\), and reduces the number of jobs from three to two. We first present the deterministic version as a warm-up and then prove the randomized theorem on the same hard instance. 

\paragraph{The hard instance.}
We use a simple product prior in which each machine independently takes one of two types, each with probability \(1/2\).
Fix  \(0<\delta < 1/4 \)  and set
\(
H=\frac4\delta.
\)
Machine \(1\) has the following two types, specifying its processing times for the two jobs:
\[
        a=(1,\; \delta),
        \qquad
        b=(1-\delta,\; 2).
\]
Machine \(2\) has two possible types,  specifying its processing times for the two jobs:
\[
        c=(2,\; H),
        \qquad
        d=(2,\; H/2).
\]
The prior is the product of these two marginals.
We show that every randomized BIC mechanism has approximation ratio at least $ \frac{4}{3+\delta}$, approaching \(4/3\) as \(\delta\to0\).

\paragraph{The first-best optimum.}
If machine \(1\) has type \(a=(1,\; \delta)\), the first-best optimum  assigns both jobs to
machine \(1\). The makespan is
\(
        1+\delta.
\)
Indeed, assigning job \(1\) to machine \(2\) already creates load \(2\), and
assigning job \(2\) to machine \(2\) creates load at least \(H/2\).
If machine \(1\) has type \(b=(1-\delta,\; 2)\), the first-best optimum assigns job \(1\) to
machine \(2\) and job \(2\) to machine \(1\). The resulting loads are \(2\) and
\(2\), so the makespan is
\(
        2.
\)
Giving both jobs to machine \(1\) gives load \(3-\delta>2\), and assigning job
\(2\) to machine \(2\) creates load at least \(H/2\).
Therefore, by Lemma~\ref{lem:randomized-opt}
\begin{equation}
\begin{aligned}
        \E[\OPT_2]
        &=
        \frac12(1+\delta)+\frac12\cdot2   \; = \; \frac{3+\delta}{2}.   \\
\end{aligned}
\label{eq:four-thirds-opt}
\end{equation}

\subsection{The switch inequality}
We now prove the common incentive constraint used in both lower bounds.
Let \(M\) be a randomized or deterministic BIC mechanism.  
For a reported type \(r\in\{a,b\}\) of machine \(1\), define the
following interim allocation probabilities, where the probability is over
machine \(2\)'s type and over the internal randomness of \(M\):
\begin{itemize}
\item 
\(
        B_r
        =
        \Prb[\text{both jobs are assigned to machine }1 \mid r],
\)
\item 
\(
        S_r
        =
        \Prb[\text{job }1\text{ is assigned to machine }1
        \text{ and job }2\text{ is assigned to machine }2 \mid r],
\)
\item
\(
        D_r
        =
        \Prb[\text{job }2\text{ is assigned to machine }2 \mid r].
\)

\end{itemize}

Clearly,  \(S_r\le D_r\). For brevity, write
\(
        B=B_a,\ B'=B_b,\
        S=S_a,\ S'=S_b,\
        D=D_a,\ D'=D_b.
\)
Then
\[
        \Prb[\text{machine }1\text{ receives job }1\mid a]=B+S,
\]
and
\[
        \Prb[\text{machine }1\text{ receives job }2\mid a]=1-D.
\]

Let \(P_a\) and \(P_b\) denote machine \(1\)'s interim expected payments when
it reports \(a\) and \(b\), respectively: 
\[
P_r
=
\E_{T_2,M}\!\left[p_1(r,T_2)\right],
\qquad r\in\{a,b\}.
\]
If the true type is \(a=(1,\; \delta)\), truthful reporting must give no larger
interim cost minus payment than reporting \(b\). Hence
\begin{equation}
(B+S)+\delta(1-D)-P_a
        \le
        (B'+S')+\delta(1-D')-P_b.
\label{eq:bic-a}
\end{equation}
If the true type is \(b=(1-\delta,\; 2)\), truthful reporting must give no larger
interim cost minus payment than reporting \(a\). Hence
\begin{equation}
(1-\delta)(B'+S')+2(1-D')-P_b
        \le
        (1-\delta)(B+S)+2(1-D)-P_a.
\label{eq:bic-b}
\end{equation}

\begin{claim}[Switch inequality]
\label{claim:switch-inequality}
Every BIC mechanism, deterministic or randomized, satisfies
\begin{equation}
\delta\bigl((B+S)-(B'+S')\bigr)
        \le
        (2-\delta)(D'-D).
\label{eq:switch-inequality}
\end{equation}
\end{claim}

\begin{proof}
 Adding  \eqref{eq:bic-a} and \eqref{eq:bic-b}, the payment terms cancel out. Rearranging the remaining terms yields the inequality:
\(
        \delta\bigl((B+S)-(B'+S')\bigr)
        \le
        (2-\delta)(D'-D) = (2-\delta) ((1-D) - (1-D')) .
\)
\end{proof}

The switch inequality captures the incentive tradeoff between the two jobs. If machine \(1\) is more likely to receive job \(1\) under \(a\) than under \(b\), BIC requires it also to receive job \(2\) more often under \(a\). This is equivalent to job \(2\) being assigned to machine \(2\) more often under \(b\), which is costly because its processing time there is at least \(H/2\).

\subsection{The Deterministic Lower Bound}

Before proving the randomized lower bound, we give a deterministic warm-up. The same construction yields an asymptotic \(4/3\) lower bound for \emph{every} deterministic BIC mechanism, improving the previous \(1.2\) bound~\cite{MS}. This class includes, in particular, all deterministic dominant-strategy truthful mechanisms. Consequently, the same \(4/3\) lower bound applies to weighted-VCG rules and, more generally, to affine minimizers with positive report-independent weights and report-independent allocation offsets. These parameters may depend on the known prior, provided they are fixed before the realized processing times are observed. Thus, even prior-dependent reweighting of VCG cannot circumvent the asymptotic \(4/3\) lower bound.   

\begin{claim}[Deterministic BIC lower bound]
\label{clm:deterministic-four-thirds}
For the product prior above, every deterministic BIC mechanism \(M\) satisfies
\[
        \frac{\E[\Cmax(M)]}{\E[\OPT_2]}
        \ge
        \frac4{3+\delta}.
\]
Hence, as \(\delta\to 0\), the lower bound approaches \(4/3\).
\end{claim}

Full proof in Appendix~\ref{sec:appendix-deterministic-four-thirds}.

\subsection{The Randomized Lower Bound}

We now prove the main theorem.

\begin{theorem}[Randomized BIC lower bound]
\label{thm:randomized-four-thirds}
For the product prior above, every randomized BIC mechanism \(M\) satisfies
\[
        \frac{\E[\Cmax(M)]}{\E[\OPT_2]}
        \ge
        \frac4{3+\delta}.
\]
Hence, as \(\delta\to 0\), the lower bound approaches \(4/3\).

\end{theorem}

\begin{proof}
Let \(M\) be any randomized BIC mechanism.
We first lower-bound the expected makespan conditional on machine \(1\)'s type.
Suppose machine \(1\) has type \(a=(1,\delta)\). If both jobs are assigned to
machine \(1\), the makespan is \(1+\delta\). If job \(2\) is assigned to
machine \(2\), the makespan is at least \(H/2\). The remaining allocation
assigns job \(1\) to machine \(2\) and job \(2\) to machine \(1\), and has
makespan \(2\). Therefore
\begin{equation}
\begin{aligned}
        \E[\Cmax(M) \mid a]
        &\ge
        B(1+\delta) + D\frac{H}{2} + 2(1-B-D)       \\
        &=
        2-(1-\delta)B+\left(\frac{H}{2}-2\right)D.
\end{aligned}
\label{eq:rand-a-cost}
\end{equation}
Similarly, suppose machine \(1\) has type \(b=(1-\delta,\; 2)\). If both jobs are
assigned to machine \(1\), the makespan is \(3-\delta\). If job \(2\) is
assigned to machine \(2\), the makespan is at least \(H/2\). The remaining
allocation assigns job \(1\) to machine \(2\) and job \(2\) to machine \(1\),
and has makespan \(2\). Therefore
\begin{equation}
\begin{aligned}
        \E[\Cmax(M) \mid b]
        &\ge
       B'(3-\delta) + D'\frac{H}{2} + 2(1-B'-D')      \\
        &=
        2+(1-\delta)B'
        +\left(\frac{H}{2}-2\right)D'.
\end{aligned}
\label{eq:rand-b-cost}
\end{equation}
Averaging \eqref{eq:rand-a-cost} and \eqref{eq:rand-b-cost} gives
\begin{equation}
\E[\Cmax(M)]
        \ge
        2+\frac{1-\delta}{2}(B'-B)
        +\frac12\left(\frac{H}{2}-2\right)(D+D').
\label{eq:rand-avg-cost}
\end{equation}
Since \(S\ge0\) and \(S'\le D'\), we have
\(
        S-S'\ge -D'.
\) 
Rearranging the switch inequality \eqref{eq:switch-inequality} gives
\begin{equation}
B'-B
        \ge
        \frac{2-\delta}{\delta}D
        -
        \frac{2-\delta}{\delta}D'
        +
        S-S'.
\label{eq:rand-b-gap-raw}
\end{equation}

Therefore
\begin{equation}
B'-B
        \ge
        \frac{2-\delta}{\delta}D
        -
        \frac2\delta D'.
\label{eq:rand-b-gap}
\end{equation}

Since \(0<\delta<1\), the coefficient of \(B'-B\) in
\eqref{eq:rand-avg-cost} is positive. Substituting \eqref{eq:rand-b-gap} into \eqref{eq:rand-avg-cost}, and using
\(H=4/\delta\) gives
\[
\begin{aligned}
        \E[\Cmax(M)]-2
        &\ge
        \frac{1-\delta}{2}
        \left(
        \frac{2-\delta}{\delta}D
        -
        \frac2\delta D'
        \right)
        +
        \frac12\left(\frac2\delta-2\right)(D+D')        \\ \\
        &=
        \frac{1-\delta}{2}
        \left(
        \frac{2-\delta}{\delta}D
        -
        \frac2\delta D'
        \right)
        +
        \frac{1-\delta}{\delta}(D+D')                    \\ \\
        &=
        \frac{1-\delta}{2}
        \left(
        \frac{2-\delta}{\delta}D
        \right)
        +
        \frac{1-\delta}{\delta}D   \;  \ge \; 0.
\end{aligned}
\]

The last inequality follows from \(D\ge0\) and \(0<\delta<1\). 

Thus
\(
        \E[\Cmax(M)]\ge2.
\)
By \eqref{eq:four-thirds-opt},
\[
        \frac{\E[\Cmax(M)]}{\E[\OPT_2]}
        \ge
        \frac{2}{(3+\delta)/2}
        =
        \frac4{3+\delta}.
\]
\end{proof}

%%%%%%%%%%%%%%%%%%%%%%%%%%%%%%%%%%%%%%%%%%%%%%%%%%%%%%%%%%%%%%%%%%%%%%%%%%
\section{The Deterministic--Randomized BIC Separation}\label{sec:det-rand-separation}
The preceding lower bound leaves open a natural question: Can allowing randomization strictly reduce the minimum expected makespan achievable under BIC? In the prior-free setting, Nisan and Ronen showed that no deterministic truthful mechanism for two machines can achieve an approximation ratio below \(2\), while also giving a randomized truthful \(7/4\)-approximation~\cite{NR}. Our next result asks whether an analogous separation occurs within the BIC class. 
This question parallels the familiar role of lotteries in multidimensional mechanism design. In multidimensional revenue maximization, randomized allocations can be necessary for optimal revenue  and can strictly outperform deterministic BIC mechanisms~\cite{MV,BCKW}.  
We answer the question affirmatively by exhibiting a two-machine, two-job, two-type per machine product prior on which randomized BIC strictly outperforms deterministic BIC. Thus, randomization can strictly reduce the minimum expected makespan achievable subject to BIC.

Let \(\OPT_{\DetBIC}\) denote the best  expected makespan achievable by
deterministic BIC mechanisms, and let \(\OPT_{\RandBIC}\) denote the best 
expected makespan achievable by randomized BIC  mechanisms.
 
\begin{theorem}[Deterministic versus randomized BIC separation]
\label{thm:deterministic-vs-randomized}
For every \(\alpha<8/7\), there exists a two-machine, two-job, two-type-per-machine product prior such that
\[
    \frac{\OPT_{\DetBIC}}{\OPT_{\RandBIC}}
    \ge \alpha.
\]
Thus \(8/7\) is an asymptotic lower bound on the deterministic-versus-randomized BIC-integrality gap.
\end{theorem}

Full proof in Appendix~\ref{sec:appendix-deterministic-vs-randomized}.

%%%%%%%%%%%%%%%%%%%%%%%%%%%%%%%%%%%%%%%%%%%%%%%%%%%%%%%%%%%%%%%%%%%%%%%%%%%%%%%%%%

\section{Resource Augmentation}
\label{sec:resource-augmentation}
Our next result focuses on  the standard MinWork mechanism and is
motivated by the notion of competition complexity. The underlying question is in the
spirit of Bulow and Klemperer: can a simple mechanism, given additional
competition, perform as well as a more sophisticated optimal Bayesian
mechanism in the original market? In the classical single-item setting
with i.i.d. regular values, Bulow and Klemperer show that the Vickrey
auction with one additional bidder obtains at least as much expected
revenue as the optimal revenue Bayesian auction with the original
bidders~\cite{BK}. Eden, Feldman, Friedler, Talgam-Cohen, and
Weinberg~\cite{EFFTW} extended this perspective to multidimensional
environments and introduced the notion of \emph{competition complexity},
measuring how many additional bidders are needed for a simple mechanism
to outperform the optimal BIC mechanism with the original bidders.

Our scheduling question is analogous, although our benchmark is stronger: we compare against the first-best optimum, rather than the second-best optimal BIC benchmark, in the original \(m\)-machine market.
We sample \(r\) independent processing-time profiles from the original
prior, calling each draw a \emph{layer}, and run MinWork on the resulting
\(mr\) destination slots, where the number of jobs remains unchanged. 
Thus \(r\) is the total number of layers, and
the augmentation adds \(r-1\) layers relative to the original market.
Let \(\MW^{(r)}\) denote MinWork in this replicated environment, with
ties broken uniformly at random among minimizing slots.

We ask how large \(r\) must be for MinWork to achieve a
\((1+\varepsilon)\)-approximation to the first-best optimum in the
original \(m\)-machine market. We show that the answer is governed by a pairwise
destination-collision parameter \(\Delta_r\). If there exists an  constant \(C_{\mathrm{ac}}\), independent of
\(m\) and \(r\), such that
\(
\Delta_r\le C_{\mathrm{ac}}
\ \text{for all }r\ge1,
\)
then \(O(m/\varepsilon)\) total layers suffice.
\subsection{Profile resampling and the original benchmark}
\label{sec:profile-resampling}
Fix an original market with \(m\) machines and \(n\) jobs. We use
\(i\in[m]\) for a machine index and \(j\in[n]\) for a job index. A
\emph{processing-time profile} is the complete matrix
\(
T=(T_{ij})_{i\in[m],\,j\in[n]}\sim F,
\)
where \(T_{ij}\ge 0\) is the processing time of job \(j\) on machine
\(i\), and \(F\) is an arbitrary joint distribution on
\(\mathbb{R}_+^{m\times n}\). We impose no independence across machines
or jobs within a profile.

For \(r\ge1\), the \emph{\(r\)-layer profile-resampling environment}
consists of \(r\) independent complete draws
\(
T^{(1)},\ldots,T^{(r)}
\overset{\mathrm{i.i.d.}}{\sim}F,
\)
where \(T^{(a)}=(T_1^{(a)},\ldots,T_m^{(a)})\), and  \(T_i^{(a)}=(T_{ij}^{(a)})_{j\in[n]}\) is the processing-time row
of machine coordinate \(i\) in layer \(a\). Each layer preserves the full
dependence structure encoded by \(F\), while different layers are
independent. Thus we resample entire processing-time profiles across all machines, rather
than machine rows independently. Here \(r\) counts total layers, so the
augmentation adds \(r-1\) layers, or \(m(r-1)\) strategic machines.

A \emph{destination slot} is a pair \(s=(i,a)\), where
\(i\in[m]\) is a machine coordinate and \(a\in[r]\) is a layer. The set
of destination slots is
\(
\mathcal S_r=[m]\times[r].
\)
Each destination slot \((i,a)\) has processing-time vector \(T_i^{(a)}=(T_{ij}^{(a)})_{j=1}^n\). We model every
destination slot as a distinct, separately owned strategic machine, so
the replicated market contains \(mr\) distinct strategic machines.

\paragraph{Tie-breaking convention.}
Throughout this section, ties are broken uniformly at random among
minimizing destination slots, independently across jobs and using
randomness independent of the reports. Equivalently, for each job we
may draw an independent uniform random ordering of the destination
slots before reports are submitted and select the first minimizing slot
in that ordering. When \(r=1\), the destination slots are the original
machines, so this reduces to uniform tie-breaking among the minimizing
machines.

MinWork assigns each job to a destination slot of minimum reported
processing time. Conditional on every realization of the tie-breaking
randomness, this is a deterministic total-work-minimizing allocation rule
with report-independent tie-breaking. Hence, for every realization of the tie-breaking randomness, procurement-VCG payments make truthful reporting a dominant strategy for every destination-slot agent~\cite{NR}. Therefore, the randomized mechanism is universally  truthful. This incentive guarantee is pointwise and
requires no distributional independence.

We henceforth evaluate the mechanism under truthful reporting. For each
job \(j\), define its random \emph{selected minimum processing time}
\begin{equation}
\label{eq:tau-r-definition}
\tau_j^{(r)}
=
\min_{i\in[m],\,a\in[r]} T_{ij}^{(a)},
\end{equation}
and let \(R_j^{(r)}\in\mathcal S_r\) denote its random destination under
MinWork, including the tie-breaking randomness. Write 
\(
\tau^{(r)}
=
(\tau_1^{(r)},\ldots,\tau_n^{(r)}),
\) 
and 
\(
R^{(r)}
=
(R_1^{(r)},\ldots,R_n^{(r)}).
\)
For every realization of the sampled profiles and tie-breaking
randomness, if \(R_j^{(r)}=s=(i,a)\), then \(s\) is a minimizing slot
for job \(j\). Hence the realized processing time of job \(j\) is the
realized value of \(\tau_j^{(r)}\). Consequently, the makespan of
\(\MW^{(r)}\) is the random variable
\begin{equation}
\label{eq:resampled-makespan}
\Cmax(\MW^{(r)})
=
\max_{s\in\mathcal S_r}
\sum_{j=1}^n
\tau_j^{(r)}
\one\{R_j^{(r)}=s\}.
\end{equation}

%%%%%%%%%%%%%%%%%%%%%%%%%%%%%

Our benchmark is the expected ex-post first-best makespan in the original
\(m\)-machine market,
\(
\E_{T\sim F}[\OPT_m(T)],
\)
rather than the optimum in the augmented \(mr\)-slot market.

For a nonnegative vector \(u=(u_1,\ldots,u_n)\), define
\begin{equation}
\label{eq:Am-definition}
A_m(u)
=
\max\left\{
\max_{j\in[n]}u_j,\;
\frac{1}{m}\sum_{j=1}^n u_j
\right\}.
\end{equation}
For a realized original profile \(t\), let
\(
\tau_j(t)=\min_{i\in[m]} t_{ij},
\ 
\tau(t)=(\tau_1(t),\ldots,\tau_n(t)).
\)

The quantity \(A_m(\tau(t))\) is a pointwise lower bound on the
first-best makespan. Indeed, under any allocation \(x\), job \(j\)
requires at least \(\tau_j(t)\) units of processing time on the machine
to which it is assigned, and hence
\(
\Cmax(x,t)\ge \max_{j\in[n]}\tau_j(t).
\)
Moreover, the total work performed by \(x\) is at least
\(\sum_j\tau_j(t)\). Since this work is distributed over \(m\) machines,
\(
\Cmax(x,t)
\ge
\frac{1}{m}\sum_{j=1}^n\tau_j(t).
\)
Taking the minimum over all allocations therefore gives
\begin{equation}
\label{eq:Am-lower-bound}
A_m(\tau(t))
\le
\OPT_m(t)
\qquad
\text{for every realized profile }t.
\end{equation}
Equivalently,
\(
A_m(\tau(T))
\le
\OPT_m(T).
\)

%%%%%%%%%%%%%%%%%%%%%%%%%
\paragraph{Pairwise destination anti-concentration.}
\label{par:destination-anti-concentration}

We next separate the two ingredients that determine the makespan of
MinWork. The vector \(\tau^{(r)}\) determines the processing times incurred
by the jobs, while their destinations determine how these processing times
are distributed across slots. We measure the latter through pairwise
destination collisions.
\pagebreak 
\begin{definition}[Pairwise destination-collision parameter]
\label{def:delta-r}
For \(n\ge2\), let \(\Delta_r\in[0,mr]\) be any number such that, for every
pair of distinct jobs \(j\ne j'\) and for  every realization
\(u\) of \(\tau^{(r)}\),
\begin{equation}
\label{eq:delta-r}
\Prb\!\left[
R_j^{(r)}=R_{j'}^{(r)}
\,\middle|\,
\tau^{(r)}=u
\right]
\le
\frac{\Delta_r}{mr}.
\end{equation}
\end{definition}

Two jobs collide only if they are sent to the same machine in the same layer. 
Conditioning on \(\tau^{(r)}\) fixes the processing time incurred by each
job under MinWork, leaving only the tie-breaking randomness. 
A valid collision parameter always exists: since conditional collision
probabilities are at most \(1\), the choice \(\Delta_r=mr\) is always
valid.~\footnote{For example, consider two machines and two jobs
(\(m=2\), \(r=1\)), and write
\(\tau=\tau^{(1)}\) and \(R_j=R_j^{(1)}\).
Suppose the processing-time matrix is, with equal probability,
\(
\begin{pmatrix}
1&1\\
1&1
\end{pmatrix}
\)
or
\(
\begin{pmatrix}
1&3\\
2&2
\end{pmatrix},
\)
where rows correspond to machines and columns to jobs.
In the first case, \(\tau=(1,1)\), and independent uniform tie-breaking
gives
\(
\Prb[R_1=R_2\mid\tau=(1,1)]=1/2.
\)
In the second case, \(\tau=(1,2)\), and the two jobs uniquely minimize on
different machines, so
\(
\Prb[R_1=R_2\mid\tau=(1,2)]=0.
\)
Thus, for every realization \(u\) of \(\tau\),
\(
\Prb[R_1=R_2\mid\tau=u]\le\frac12.
\)
Since \(m=2\) and \(r=1\), \(\Delta_1=1\) is a valid and tight choice. The key point is that \(\tau_j\) records only the minimum processing time of job \(j\), not the identity of a machine attaining that minimum. For example, \(\tau_j=1\) means only that the minimum available processing time for job \(j\) is \(1\); by itself, it does not specify the job's destination.}

Under the independence and machine-symmetry assumptions of Giannakopoulos and Kyropoulou~\cite{GK}, applying the profile-resampling construction to their setting yields conditionally independent uniform destinations over the \(mr\) slots.
Hence \(\Delta_r=1\) is a valid and tight choice.

By contrast, in the two-machine hard instance  of
Section~\ref{sec:four-thirds-randomized-bic-lb}, both jobs minimize only
among the \(r\) replicas of machine~1. For some realizations of
\(\tau^{(r)}\), they tie over the same \(r\) slots and therefore collide
with probability \(1/r\). Thus
\(
\frac{1}{r}=\frac{2}{2r},
\)
and \(\Delta_r=2=m\).
More generally, every profile-resampling environment admits a collision
parameter \(\Delta_r\), and
Theorem~\ref{thm:competition-complexity} applies without any further
independence assumption.
%%%%%%%%%%%%%%%%%%%%%%%%%%%%%%%%%%%%%%%%%%%

\begin{theorem}
\label{thm:competition-complexity}
For every \(r\ge1\) and every \(\Delta_r\) satisfying
Definition~\ref{def:delta-r},
\[
\E[\Cmax(\MW^{(r)})]
\le
\left(
1+\frac{\Delta_r(m-1)}{2r}
\right)
\E_{T\sim F}[\OPT_m(T)].
\]
\end{theorem}

\begin{proof}
Fix a realization \(u\) of \(\tau^{(r)}\) for which the collision bounds
in Definition~\ref{def:delta-r} hold for every pair of distinct jobs.
Let
\(
A=A_m(u).
\)
If \(A=0\), then \(u_j=0\) for every job \(j\), and the conditional
makespan is zero. Assume \(A>0\), and define
\(
y_j=\frac{u_j}{A}.
\)
By the definition of \(A_m\),
\(
0\le y_j\le1
 \ \text{and} \ 
\sum_{j=1}^n y_j\le m.
\)
Hence
\[
y\in\mathcal P_{n,m}
:=
\left\{
z\in[0,1]^n:
\sum_{j=1}^n z_j\le m
\right\}.
\]

For \(z\in\mathcal P_{n,m}\), define
\[
G_u^{(r)}(z)
=
\E\!\left[
\max_{s\in\mathcal S_r}
\sum_{j=1}^n
z_j\one\{R_j^{(r)}=s\}
\,\middle|\,
\tau^{(r)}=u
\right].
\]
Thus \(G_u^{(r)}(z)\) has the following interpretation. Conditioning on
\(\tau^{(r)}=u\) fixes the conditional distribution of the MinWork
destinations \(R^{(r)}\). Keeping this routing distribution fixed, assign job \(j\) size \(z_j\), draw the destinations from this conditional distribution, and compute the maximum total size assigned to any destination slot.
The expectation of this maximum load is \(G_u^{(r)}(z)\).

For every fixed realization of the destinations \(R^{(r)}\), the function
\(
\max_{s\in\mathcal S_r}
\sum_{j=1}^n
z_j\one\{R_j^{(r)}=s\}
\)
is the maximum of finitely many linear functions of \(z\), one for each
destination slot \(s\in\mathcal S_r\), and is therefore convex in \(z\).
Since \(R^{(r)}\) has only finitely many possible realizations,
\(G_u^{(r)}\) is a convex combination of such convex functions. Hence
\(G_u^{(r)}\) is convex.

Since \(m\) is an integer, the vertices of \(\mathcal P_{n,m}\) are the
indicator vectors \(\mathbf 1_S\) of sets \(S\subseteq[n]\) with
\(|S|\le m\). Every \(z\in\mathcal P_{n,m}\) is a convex combination of
such vertices. Therefore, by convexity,
\(
G_u^{(r)}(z)
\le
\max_{\substack{S\subseteq[n]\\ |S|\le m}}
G_u^{(r)}(\mathbf 1_S).
\)

Fix \(S\subseteq[n]\) with \(|S|\le m\). For each destination slot
\(s\in\mathcal S_r\), let
\(
N_s
=
\sum_{j\in S}\one\{R_j^{(r)}=s\}
\ \text{and} \ 
Q=\max_{s\in\mathcal S_r}N_s.
\)
Thus \(Q\) is the maximum number of jobs in \(S\) assigned to a single
destination slot. For every realized routing,
\(
\sum_{\substack{j<j'\\ j,j'\in S}}
\one\{R_j^{(r)}=R_{j'}^{(r)}\}
=
\sum_{s\in\mathcal S_r}\binom{N_s}{2}
\ge
\binom Q2
\ge Q-1.
\)
Hence
\(
Q
\le
1+
\sum_{\substack{j<j'\\ j,j'\in S}}
\one\{R_j^{(r)}=R_{j'}^{(r)}\}.
\)
Since \(z=\mathbf 1_S\) gives unit size to the jobs in \(S\) and zero
size to all other jobs,
\(
G_u^{(r)}(\mathbf 1_S)
=
\E\!\left[
Q
\,\middle|\,
\tau^{(r)}=u
\right]
\le
1+
\E\!\left[
\sum_{\substack{j<j'\\ j,j'\in S}}
\one\{R_j^{(r)}=R_{j'}^{(r)}\}
\,\middle|\,
\tau^{(r)}=u
\right]
\\
=
1+
\sum_{\substack{j<j'\\ j,j'\in S}}
\Prb\!\left[
R_j^{(r)}=R_{j'}^{(r)}
\,\middle|\,
\tau^{(r)}=u
\right].
\)

By Definition~\ref{def:delta-r},
\(
G_u^{(r)}(\mathbf 1_S)
\le
1+
\binom{|S|}{2}\frac{\Delta_r}{mr}
\le
1+
\binom m2\frac{\Delta_r}{mr}
=
1+\frac{\Delta_r(m-1)}{2r}.
\)
Thus for  every \(z\in\mathcal P_{n,m}\) we have,  
\(
G_u^{(r)}(z)
\le
1+\frac{\Delta_r(m-1)}{2r}.
\)
In particular,
\(
G_u^{(r)}(y)
\le
1+\frac{\Delta_r(m-1)}{2r}.
\)
Finally, conditional on \(\tau^{(r)}=u\), the actual processing time of
job \(j\) under MinWork is
\(
u_j=A\cdot y_j.
\)
Therefore
\begin{align*}
\E\!\left[
\Cmax(\MW^{(r)})
\,\middle|\,
\tau^{(r)}=u
\right]
\; = \;
A \cdot \,G_u^{(r)}(y)
\; \le \; 
A_m(u) \cdot 
\left(
1+\frac{\Delta_r(m-1)}{2r}
\right).
\end{align*}
Taking expectations, 
\[
\E[\Cmax(\MW^{(r)})]
\le
\left(
1+\frac{\Delta_r(m-1)}{2r}
\right)
\E\!\left[A_m(\tau^{(r)})\right].
\]
Now, Couple the first layer with the original profile, \(T^{(1)}=T\).
Since
\(
\tau_j^{(r)}
=
\min_{i,a} T_{ij}^{(a)}
\le
\min_i T_{ij}^{(1)}
=
\tau_j(T),
\)
coordinatewise monotonicity of \(A_m\) gives
\(
A_m(\tau^{(r)})
\le
A_m(\tau(T)).
\) Finally, by~\eqref{eq:Am-lower-bound}, 

\[
\E[\Cmax(\MW^{(r)})]
\le
\left(
1+\frac{\Delta_r(m-1)}{2r}
\right)
\E_{T\sim F}[\OPT_m(T)].
\]
\end{proof}  % Theorem 3 

\subsection{Consequences and Examples}
\label{sec:uniform-collision-bounds}

Theorem~\ref{thm:competition-complexity} becomes especially simple when \(\Delta_r\) is  bounded. In this subsection, we derive the corresponding number of layers needed for a \((1+\varepsilon)\)-approximation.

\begin{proposition}
\label{prop:uniform-collision-upper}
Suppose
\(
\Delta_r \le C_{\mathrm{ac}}
\ \text{for all } r\ge1.
\)
Then, for every \(\varepsilon>0\),
\(
r\ge
\left\lceil
\frac{C_{\mathrm{ac}}(m-1)}{2\varepsilon}
\right\rceil
\)
total layers suffice to guarantee
\(
\E[\Cmax(\MW^{(r)})]
\le
(1+\varepsilon)\E_{T\sim F}[\OPT_m(T)].
\)
In particular, if \(C_{\mathrm{ac}}\) 
is  independent of $m$ across a class  of priors, then
\(
O\!\left(\frac{m}{\varepsilon}\right)
\)
total layers suffice for a \((1+\varepsilon)\)-approximation throughout the family.
\end{proposition}

\begin{proof}
By Theorem~\ref{thm:competition-complexity},
\(
\E[\Cmax(\MW^{(r)})]
\le
\left(
1+\frac{\Delta_r(m-1)}{2r}
\right)
\E_{T\sim F}[\OPT_m(T)].
\)
Since \(\Delta_r\le C_{\mathrm{ac}}\), we have
\(
\frac{\Delta_r(m-1)}{2r}
\le
\frac{C_{\mathrm{ac}}(m-1)}{2r}.
\)
Therefore, if
\(
r\ge
\left\lceil
\frac{C_{\mathrm{ac}}(m-1)}{2\varepsilon}
\right\rceil,
\)
then
\(
\frac{\Delta_r(m-1)}{2r}\le\varepsilon,
\)
and hence
\(
\E[\Cmax(\MW^{(r)})]
\le
(1+\varepsilon)\E_{T\sim F}[\OPT_m(T)].
\)

If $C_{\mathrm{ac}}$ is independent of $m$, 
then choosing
$ \left\lceil \frac{C_{\mathrm{ac}}(m-1)}{2\varepsilon} \right\rceil =O\left(\frac{m}{\varepsilon}\right) $
total layers suffices for a \((1+\varepsilon)\)-approximation throughout the class of priors.
\end{proof}

%%%%%%%%%%%%%%%%%%%%%%%%%%%%%%%%%%
%%%%%%%%%%%%%%%%%%%%%%%%%%%%%%%%%%%

%%%%%%%%%%%%%%%%%%%%%%%%%%%%%%%%%%%%%%%%%%%%%%%%%%%%%%%%%%%%%%%%%%

\paragraph{Examples.}
We next give three examples illustrating
Proposition~\ref{prop:uniform-collision-upper}.

\begin{itemize}[leftmargin=*]

\item \textbf{No collisions.}
Suppose that, in every realized layer \(a\), distinct jobs have disjoint
sets of minimizing machine coordinates:
\[
\operatorname*{arg\,min}_{i} T_{ij}^{(a)}
\cap
\operatorname*{arg\,min}_{i} T_{ij'}^{(a)}
=
\varnothing
\qquad
\text{for all } j\ne j'.
\]
Recall that a collision can occur only if two jobs choose the same machine coordinate in the same layer, which is ruled out by the disjoint-minimizer condition above. Hence, for every realization \(u\) of \(\tau^{(r)}\) and for any distinct jobs \(j\ne j'\),
\(
\Prb\!\left[
R_j^{(r)}=R_{j'}^{(r)}
\,\middle|\,
\tau^{(r) } =u
\right]
=0,
\) so \(\Delta_r=0\) for every \(r\).

In fact, no augmentation is needed. When \(r=1\), MinWork assigns
different jobs to different machines and therefore has makespan
\(
\max_j \tau_j(T).
\)
This is a lower bound on \(\OPT_m(T)\), and MinWork attains it. Hence
MinWork is first-best for every realization.

\item \textbf{Independent and machine-identical processing times.}
Under the assumptions of Giannakopoulos and Kyropoulou~\cite{GK}, the processing times
$
\{T_{ij}\}_{i\in[m],\,j\in[n]}
$
are mutually independent, and for each job \(j\) there is a distribution
\(F_j\) such that
$
T_{ij}\sim F_j
\ 
\text{for every machine } i.
$
Thus, for each fixed job, the processing time is identically distributed
across machines, while the distribution may differ across jobs.
Under our profile-resampling construction, the \(mr\) destination-slot
processing times of each job are therefore identically distributed and
independent, and the corresponding vectors are independent across jobs.
Together with independent uniform tie-breaking, this implies that,
conditional on \(\tau^{(r)}\), the destinations selected by distinct jobs
are independent and uniform over the \(mr\) destination slots.  Therefore, for every realization \(u\) of \(\tau^{(r)}\) and for any  distinct jobs \(j\ne j'\),
\[
\Prb\!\left[
R_j^{(r)}=R_{j'}^{(r)}
\,\middle|\,
\tau^{(r) }= u
\right]
=
\frac{1}{mr}.
\]
Thus for \(n\ge2\),  \(\Delta_r=1\) is a valid and tight choice for every \(r\).
The same constant works for all \(m\), so
Proposition~\ref{prop:uniform-collision-upper} gives
\(O(m/\varepsilon)\) total layers for a
\((1+\varepsilon)\)-approximation.

\item \textbf{A common fast-machine pool.}
Suppose \(m\) is even and the machines are partitioned in advance into two fixed pools of size \(m/2\). In each layer, one of the two pools is chosen uniformly at random, independently across layers, to be fast for all jobs. Machines in the fast pool have processing time 1, while machines in the other pool have processing time 2.
This models, for example, a shared caching or data-locality effect:
in each layer, one pool has the relevant data already cached and can
therefore process all jobs faster.

Each layer contributes \(m/2\) minimizing destination slots, so across
\(r\) layers every job has exactly \(mr/2\) minimizing slots. Conditional
on any realization of the fast pools, independent uniform tie-breaking
therefore gives
\[
\Prb\!\left[
R_j^{(r)}=R_{j'}^{(r)}
\,\middle|\,
\text{fast-pool realizations}
\right]
=
\frac{1}{mr/2}
=
\frac{2}{mr}.
\]
This value is the same for every realization of the fast pools.
Thus for \(n\ge2\), \(\Delta_r=2\) is a valid and tight choice for every \(r\).
Since the same constant works for every even \(m\),
Proposition~\ref{prop:uniform-collision-upper} gives
\(O(m/\varepsilon)\) total layers for a
\((1+\varepsilon)\)-approximation.

\end{itemize}

The next result shows that this \(m/\varepsilon\) dependence is tight. In particular, \(\Omega(m/\varepsilon)\) layers may be necessary
even in the independent-uniform case above, where \(\Delta_r=1\).
%%%%%%%%%%%%%%%%%%############################################
\begin{proposition}
\label{prop:all-unit-lower}
Let \(n=m\ge 2\), and \(T_{ij}\equiv1\) for every \(i,j\).
Then \(\OPT_m=1\) and, under the independent uniform tie-breaking
convention, \(\Delta_r=1\) for every \(r\). If
$
\E[\Cmax(\MW^{(r)})]\le1+\varepsilon,$ and 
$0<\varepsilon<\frac12,
$
then

$$
r
\ge
\frac{m-1}{2[-\log(1-\varepsilon)]}
>
\frac{m-1}{4\varepsilon}.
$$

In particular, \(r=\Omega(m/\varepsilon)\).

\end{proposition}

Full proof in Appendix~\ref{sec:appendix-BK-LB}.

%%%%%%%%%%%%%%%%%%%%%%%%%%%%%%%%%%
%%%%%%%%%%%%%%%%%%%%%%%%%%%%%%%%%%%

\subsection{The Unaugmented Case: A General Upper Bound}
\label{sec:r-one-gk}
Observe that \(\Delta_1\) is defined for every prior, without any symmetry or independence assumption. Thus, setting \(r=1\),  Theorem~\ref{thm:competition-complexity}   gives the following prior-dependent performance bound for ordinary MinWork:
$ \E[\Cmax(\MW)] \le \left(1+\frac{\Delta_1(m-1)}{2}\right)\E[\OPT_m]. $
MinWork does not need to know the prior or \(\Delta_1\); whenever \(\Delta_1\) can be upper bounded, the theorem gives an explicit guarantee. Because it depends only on pairwise destination collisions, the bound applies more broadly than the Giannakopoulos--Kyropoulou analysis~\cite{GK}, which assumes stronger symmetry and independence. For \(m=2\) and \(\Delta_1=1\), Theorem~\ref{thm:competition-complexity} gives a \(3/2\)-approximation. This is tight: on the all-unit instance with two jobs and two machines, which is the canonical balls-into-bins instance underlying the Giannakopoulos--Kyropoulou analysis~\cite{GK}, MinWork has expected makespan \(3/2\), while the first-best makespan is \(1\).  For larger \(m\), however, the Giannakopoulos--Kyropoulou bound can be substantially sharper under their symmetric independent setting.

%%%%%%%%%%%%%%%%%%%%%%%%%%%%%%%%%%%%%%%%%%%%%%%%%%%%%

\section{Conclusion}
This paper separates three questions in Bayesian scheduling. First, randomized BIC mechanisms can remain bounded away from the first best: we prove an asymptotic \(4/3\) lower bound. Second, randomization can strictly improve BIC performance, yielding an asymptotic lower bound of \(8/7\) on the deterministic-to-randomized BIC gap. Thus, optimal prior-dependent BIC scheduling can require randomization. Third, we show that resource augmentation can make MinWork nearly first best. In the profile-resampling model,
\(
\E[\Cmax(\MW^{(r)})]
\le
\left(1+\frac{\Delta_r(m-1)}{2r}\right)\E[\OPT_m].
\)
Hence, when \(\Delta_r\) is bounded by a constant independent of \(m\) and \(r\), \(O(m/\varepsilon)\) sampled layers suffice for a \((1+\varepsilon)\)-approximation (equivalently, \(O(m/\varepsilon)\) sampled replicas per machine suffice), and this dependence is asymptotically tight in the worst case. As a by-product, setting \(r=1\) yields a general upper bound for ordinary MinWork on every prior.

\sloppy

\appendix
\section{Deferred Proofs}
\label{sec:appendix}

\subsection{Proof of Claim~\ref{clm:deterministic-four-thirds}}
\label{sec:appendix-deterministic-four-thirds}
\begin{proof}
Let \(M\) be a deterministic BIC mechanism. 
First suppose that \(D>0\) or \(D'>0\). Since \(M\) is deterministic and
machine \(2\)'s two types are equally likely, this means that on at least one
of the four type profiles, job \(2\) is assigned to machine \(2\). On that
profile, the makespan is at least \(H/2\). Since each profile has probability
\(1/4\),
\(
        \E[\Cmax(M)]
        \ge
        \frac14\cdot\frac{H}{2}
        =
        \frac H8.
\)
Because \(H=4/\delta\) and \(\delta < 1/4\), we have \(H\ge 16\). Therefore
\(
        \E[\Cmax(M)]\ge2.
\)
It remains to consider the case \(D=D'=0\). Then job \(2\) is assigned on every type profile
to machine \(1\), and therefore \(S=S'=0\). The switch inequality becomes
\(
        \delta(B-B')\le0,
\)
so
\begin{equation}
B\le B'.
\label{eq:det-monotonicity}
\end{equation}

Conditional on machine \(1\) having type \(a=(1,\delta)\), the mechanism either
assigns both jobs to machine \(1\), which has makespan \(1+\delta\), or assigns
job \(1\) to machine \(2\) and job \(2\) to machine \(1\), which has makespan
\(2\). Hence
\begin{equation}
\E[\Cmax(M) \mid a]
        =
        B(1+\delta)+(1-B)2
        =
        2-(1-\delta)B.
\label{eq:det-a-cost}
\end{equation}
Conditional on machine \(1\) having type \(b=(1-\delta,2)\), assigning both
jobs to machine \(1\) has makespan \(3-\delta\), while assigning job \(1\) to
machine \(2\) and job \(2\) to machine \(1\) has makespan \(2\). Hence
\begin{equation}
\E[\Cmax(M)\mid b\;]
        =
        B'(3-\delta)+(1-B')2
        =
        2+(1-\delta)B'.
\label{eq:det-b-cost}
\end{equation}
Averaging \eqref{eq:det-a-cost} and \eqref{eq:det-b-cost}, and using
\eqref{eq:det-monotonicity}, gives
\(
        \E[\Cmax(M)]
        =
        2+\frac{1-\delta}{2}(B'-B)  \; \ge  \; 2. 
\)
Thus in all cases \(\E[\Cmax(M)]\ge2\). By \eqref{eq:four-thirds-opt},
\[
        \frac{\E[\Cmax(M)]}{\E[\OPT_2]}
        \ge
        \frac{2}{(3+\delta)/2}
        =
        \frac4{3+\delta}.
\]
\end{proof}

\subsection{Proof of Theorem~\ref{thm:deterministic-vs-randomized}}
\label{sec:appendix-deterministic-vs-randomized}
\begin{proof}
For a type \(r\) of a given machine, let \(q_r\) denote its interim allocation vector, whose \(j\)-th coordinate is the probability that job \(j\) is assigned to that machine. Let \(P_r\) denote its interim expected payment. We consider unrelated-machines makespan minimization with two machines and
two jobs.  Each machine has two possible types, each occurring with
probability \(1/2\), independently across machines.
We denote allocations by \(11,12,21,22\), where the first digit is the machine receiving job 1 and the second the machine receiving job 2.
Fix an integer \(H\ge 2\).  Machine 1 has two possible types, each with probability 1/2:
\[
    A=(H,\; H+1),
    \qquad
    B=(3H+1,\;H).
\]
Machine 2 has two possible types, each with probability 1/2:
\[
    C=(4H+1,\;3H+1),
    \qquad
    D=(2,\;2H+1).
\]
The prior is the product of these two marginals.
The makespan table is:
\[
\begin{array}{c|rrrr}
\text{Profile} & 11\;  & 12\;  & 21\;  & 22\; \\
\hline
(A,C) & 2H+1 & 3H+1 & 4H+1 & 7H+2\\
(A,D) & 2H+1 & 2H+1 & H+1   & 2H+3\\
(B,C) & 4H+1 & 3H+1 & 4H+1 & 7H+2\\
(B,D) & 4H+1 & 3H+1 & H     & 2H+3
\end{array}
\]

Thus the unique profile-wise  first-best optimum  is
\(
(A,C)\mapsto 11,\;
(A,D)\mapsto 21,\;
(B,C)\mapsto 12,\;
(B,D)\mapsto 21.
\)
Therefore by Lemma~\ref{lem:randomized-opt}
\[
    \mathbb E[\OPT_2]
    = \frac{(2H+1)+(H+1)+(3H+1)+H}{4} = 
    \frac{7H+3}{4}.
\]
\paragraph{The deterministic BIC optimum.}
We  show that the first-best optimum is not BIC.
Under the first-best, machine 1's interim allocations are
\(
    q_A=\left(\frac12,\; 1\right),
    \;
    q_B=\left(\frac12,\; \frac12\right).
\)

The BIC constraints
for machine 1 require
\[
    P_A-P_B
    \ge
    A\cdot(q_A-q_B)
    =
    (H,H+1)\cdot\left(0,\; \frac12\right)
    =
    \frac{H+1}{2},
\]
and also
\[
    P_A-P_B
    \le
    B\cdot(q_A-q_B)
    =
    (3H+1,H)\cdot\left(0,\; \frac12\right)
    =
    \frac H2.
\]
This is impossible.  Therefore the first-best is not BIC.
Moreover, as seen from the table above, on each profile, every deterministic allocation other than the first-best increases the makespan by at least \(H\). Hence any deterministic rule with total cost strictly  below
$ (7H+3)+H
    = 8H+3 $
is non-BIC.
 Therefore every deterministic
BIC mechanism has total cost at least \(8H+3\), and
\(
    \OPT_{\DetBIC}
    \ge
    \frac{8H+3}{4}.
\)
This lower bound is tight.  

Consider the deterministic rule (which changes only the allocation for $(B, C)$)): 
\(
(A,C)\mapsto 11,\;
(A,D)\mapsto 21,\;
(B,C)\mapsto 11,\;
(B,D)\mapsto 21.
\)
Its total cost is
\(
    (2H+1)+(H+1)+(4H+1)+H
    =
    8H+3.
\)
We show that this deterministic allocation rule is BIC.  Clearly, for machine 1, we have 
\(
    q_A=q_B=\left(\frac12,\; 1\right),
\)
hence its BIC constraints are satisfied by choosing \(P_A=P_B\).

For machine 2, we have 
\(
    q_C=(0,\; 0),
    \; 
    q_D=(1,\; 0).
\)
Thus
\(
    q_C-q_D=(-1,\; 0).
\)
The BIC constraints for machine 2 require
\(
    C\cdot(q_C-q_D)
    \le
    P_C-P_D
    \le
    D\cdot(q_C-q_D).
\)
That is,
\(
    -(4H+1)
    \le
    P_C-P_D
    \le
    -2.
\)
This interval is nonempty.  For example, choose
\(
    P_C-P_D=-2.
\)
Thus the rule is deterministic BIC.
Consequently,
\(
    \OPT_{\DetBIC}
    =
    \frac{8H+3}{4}.
\)

\paragraph{A randomized BIC mechanism}

Let 
\(
    \delta=\frac{1}{2H+2}.
\)
Now consider the following randomized mechanism:
\(
(A,C)\mapsto 11, \; (B,C)\mapsto 12,
\;
(B,D)\mapsto 21, 
\) and 
\[
(A,D)\mapsto
\begin{cases}
21 & \text{with probability }1-\delta,\\
12 & \text{with probability }\delta.
\end{cases}
\]

This mechanism differs from the  first-best optimum only at profile
\((A,D)\), where allocation \(12\) has cost \(2H+1\) instead of the optimal
cost \(H+1\).  The cost increase at that profile is \(H\), and it is incurred
with probability \(\delta\).  Therefore the sum of the four conditional expected profile costs is
\(
    7H+3+\delta H
    =
    7H+3+\frac{H}{2H+2}.
\)
Hence this randomized mechanism has expected makespan
\(
    \frac14
    \left(
        7H+3+\frac{H}{2H+2}
    \right).
\)
It remains to verify BIC.
For machine 1, the interim allocations are
\(
    q_A
    =
    \left(\frac{1+\delta}{2}, \; 1-\frac{\delta}{2}\right),
    \; 
    q_B
    =
    \left(\frac12,\frac12\right).
\)
Therefore
\(
    q_A-q_B
    =
    \left(\frac{\delta}{2}, \;\frac{1-\delta}{2}\right).
\)

Machine 1's BIC constraints are feasible if
\(
    A\cdot(q_A-q_B)
    \le P_A - P_B \le 
    B\cdot(q_A-q_B).
\)
We have
\[
    A\cdot(q_A-q_B)
    =
    \frac{H\delta+(H+1)(1-\delta)}{2}
    =
    \frac{H+1-\delta}{2},
\]
and
\[
    B\cdot(q_A-q_B)
    =
    \frac{(3H+1)\delta+H(1-\delta)}{2}
    =
    \frac{H+(2H+1)\delta}{2}.
\]
These two quantities are equal exactly when
\(
    H+1-\delta=H+(2H+1)\delta,
\)
or equivalently
\(
    \delta=\frac{1}{2H+2}.
\)
Thus machine 1's BIC constraints bind exactly and can be satisfied by an
appropriate choice of \(P_A - P_B =  \frac{H+1-\delta}{2} \).

For machine 2, the interim allocations are
\(
    q_C=\left(0, \; \frac12\right),
    \;
    q_D=\left(1-\frac{\delta}{2},\; \frac{\delta}{2}\right).
\)
Hence
\(
    q_C-q_D
    =
    \left(
        -1+\frac{\delta}{2},
        \; \frac12-\frac{\delta}{2}
    \right).
\)
Machine 2's BIC constraints are feasible if
\[
    C\cdot(q_C-q_D)
    \le P_C - P_D\le 
    D\cdot(q_C-q_D).
\]
It is enough to verify that  \(
    C\cdot\left(-2+\delta , 1 -\delta \right)
    <  0 <  
    D\cdot\left(-2+\delta , 1 -\delta \right).
\)
Now
\[
    C\cdot\left(-2+\delta , 1 -\delta \right) 
    =
    (4H+1)(-2)+(4H+1)\delta +(3H+1)-(3H+1)\delta =   -(5H+1)+H \delta 
\]
Substituting \(
    \delta=\frac{1}{2H+2}
\) 
gives   
\(
   -(5H+1)+\frac{H}{2H+2} <  -(5H+1)+1  < 0,  
\) 
 for every $H \ge 2$. 
In addition, 
\[
    D\cdot\left(-2+\delta , 1 -\delta \right) 
    = 2(-2+\delta)+(2H+1)(1-\delta) = 2H-3-\delta(2H-1).
\]
Substituting \(
    \delta=\frac{1}{2H+2}
\) gives   
\(
2H - 3 - \frac{2H-1}{2H+2} = \frac{4H^2 - 4H -5}{2H+2} > 0  
\)
for every $H \ge 2$.
Hence machine 2's two BIC inequalities are simultaneously satisfied by choosing
\(
    P_C-P_D=0.
\)
Thus the randomized mechanism is BIC.  Consequently,
\(
    \OPT_{\RandBIC}
    \le
    \frac14
    \left(
        7H+3+\frac{H}{2H+2}
    \right).
\)

\paragraph{The integrality gap}

Combining the deterministic lower bound and the randomized upper bound gives
\[
    \frac{\OPT_{\DetBIC}}{\OPT_{\RandBIC}}
    \ge
    \frac{
        \frac{8H+3}{4}
    }{
        \frac14\left(7H+3+\frac{H}{2H+2}\right)
    }
    =
    \frac{8H+3}{7H+3+\frac{H}{2H+2}}.
\]
As \(H\to\infty\),
\(
    \; \frac{8H+3}{7H+3+\frac{H}{2H+2}}
    \longrightarrow
    \frac87.
\)
\end{proof}

\subsection{Proof of Proposition~\ref{prop:all-unit-lower}}
\label{sec:appendix-BK-LB}
\begin{proof}
The first-best optimum is \(1\), since the \(m\) unit jobs can be assigned
one per machine. In the \(r\)-layer environment, every job ties across
all \(mr\) destination slots, so independent uniform tie-breaking makes
the \(m\) destinations independent and uniform. Thus two distinct jobs
collide with probability \(1/(mr)\), and hence \(\Delta_r=1\).
Let \(Y=\Cmax(\MW^{(r)})\), the maximum occupancy of the \(mr\) slots.
The probability of no collision satisfies $
\Prb[\text{no collision}]
=
\prod_{\ell=0}^{m-1}
\left(1-\frac{\ell}{mr}\right)
\le
\exp\!\left(-\frac{m-1}{2r}\right).
$

Since \(Y\ge1\) always and \(Y\ge2\) whenever a collision occurs,
$
\E[Y]
\ge
1+\Prb[\text{collision}]
\ge
2-\exp\!\left(-\frac{m-1}{2r}\right).
$
Therefore \(\E[Y]\le1+\varepsilon\) implies
$
1-\exp\!\left(-\frac{m-1}{2r}\right)
\le\varepsilon,
$
and hence
$
r
\ge
\frac{m-1}{2[-\log(1-\varepsilon)]}.
$
Finally, for \(0<\varepsilon<1/2\), we have 
$
-\log(1-\varepsilon)
\le
\frac{\varepsilon}{1-\varepsilon}
<
2\varepsilon,
$
which yields
$
r>\frac{m-1}{4\varepsilon}.
$
\end{proof}
\end{document}